\documentclass[letter,11pt]{article}
\usepackage[utf8x]{inputenc}
\usepackage{paithan}
\usepackage[capitalise]{cleveref}

\newcommand{\uncolored}[0]{\ensuremath{Uncolored}}
\newcommand{\red}[0]{\ensuremath{Red}}
\newcommand{\blue}[0]{\ensuremath{Blue}}
\newcommand{\black}[0]{\ensuremath{Black}}
\newcommand{\white}[0]{\ensuremath{White}}
\newcommand{\btcl}[0]{\ruleset{Bounded 2-Player Constraint Logic}} 

\title{PSPACE-Complete Two-Color Placement Games}
\author{Kyle Burke, Bob Hearn}

\begin{document}

\maketitle

\begin{abstract}
  We show that three placement games, \ruleset{Col}, \ruleset{NoGo}, and \ruleset{Fjords}, are \cclass{PSPACE}-complete on planar graphs.  The hardness of \ruleset{Col} and \ruleset{Fjords} is shown via a reduction from \btcl and \ruleset{NoGo} is shown to be hard directly from \ruleset{Col}.
\end{abstract}

\section{Background}

\subsection{Combinatorial Game Theory}

\emph{Combinatorial Game Theory} is the study of games with:
\begin{itemize}
  \item Two players alternating turns,
  \item No randomness, and
  \item Perfect information for both players.
\end{itemize}

A \emph{ruleset} is a pair of functions that determines which moves each player can make from some position.  Most games in this paper use \emph{normal play} rules, meaning if a player can't make a move on their turn, they lose the game (i.e. the last player to move wins).

The two players are commonly known as \emph{Left} and \emph{Right}.  The rulesets discussed here include players assigned to different colors (\blue\ vs \red\ or \black\ vs \white).  We use the usual method of distinguishing between them: Left will play as \blue\ and \black; Right plays \red\ and \white.

For more information on combinatorial game theory, the reader is encouraged to look at \cite{WinningWays:2001} and \cite{LessonsInPlay:2007}.

\subsection{Algorithmic Combinatorial Game Theory}

\emph{Algorithmic Combinatorial Game Theory} is the application of algorithms to combinatorial games.  The difficulty of a ruleset is analyzed by studying the computational complexity of determining whether the current player has a winning strategy.  In this paper, we show that many games are \cclass{PSPACE}-complete, which means that no polynomial-time algorithm exists to determine the winnability of all positions unless such an algorithm exists for all \cclass{PSPACE} problems.  

Usually determining the winnability of a ruleset is considered as the computational problem of the same name.  We use that language here, e.g. saying \btcl\ is \cclass{PSPACE}-complete means that the associated winnability problem is \cclass{PSPACE}-complete.

All games considered in this paper exist in \cclass{PSPACE} due to the max height of the game tree being polynomial \cite{}.  Thus, by showing that any of these games are \cclass{PSPACE}-hard, we also show that they are \cclass{PSPACE}-complete.

For more on algorithmic combinatorial game theory, the reader is encouraged to reference \cite{AlgGameTheory_GONC3}.

\subsection{\btcl}

\btcl\ is a combinatorial ruleset played on a directed graph where each arc has three properties:
\begin{itemize}
  \item Color: which of the two players is allowed to flip it.
  \item Flipped: a boolean flag indicating whether it has already been flipped.  Each arc may be flipped only once.
  \item Weight: one of $\{1, 2\}$.\footnote{Warning: in \cite{DBLP:books/daglib/0023750}, these weights are denoted by blue vs. red edges.  These colors do not correspond to the identity of the player that may flip the arc.}
\end{itemize}
An orientation of the arc is legal if each vertex in the graph has at total weight of incoming edges of at least 2.  A move consists of a player choosing an arc, $(v, w)$, to flip where:
\begin{itemize}
  \item The arc is that player's color, and
  \item The arc has not yet been flipped, and
  \item Flipping the arc (meaning the graph with $(w, v)$ replacing $(v, w)$) results in a legal orientation.
\end{itemize}

The goal of \btcl\ is for Left to flip a goal edge.  If they can flip this edge, then they win the game.  Otherwise, Right wins.

\btcl\ is \cclass{PSPACE}-complete, even when:
\begin{itemize}
  \item The graph is planar, and
  \item Only six types of vertices exist in the graph.  
\end{itemize}
These six vertex types are: And, Or, Choice, Split, Variable, and Goal, named for the gadgets they represent in the proof of \btcl\ hardness \cite{DBLP:books/daglib/0023750}.  The following is a description of each of these vertices.  Diagrams for each may be found in \cite{DBLP:books/daglib/0023750}.

\begin{itemize}
  \item Variable: One of two edges (one of each color) may be flipped.  \black's edge corresponds to setting the variable to \emph{true}, \white\ sets it to \emph{false}.
  \item Goal: This is the \black\ edge that Left needs to flip to win the game.
  \item And: A vertex with two outward-oriented ``input'' edges and one inward-oriented ``output'' edge.  In order to flip the output edge, both input edges must first be flipped.
  \item Or: Another vertex with two inputs and one output, but here only one of the inputs must be flipped in order for the output to be flipped.
  \item Choice: One input edge which, when flipped to orient inwards, means one of two output edges may be flipped to orient outwards.
  \item Split: One input edge which, when flipped to orient inwards, allows both output edges to be flipped orienting outwards.
\end{itemize}

In order to use \btcl\ as the source problem for a proof of \cclass{PSPACE}-hardness, it is sufficient to show that gadgets that simulate each of the six \btcl\ vertex types.  We use this to reduce directly to \ruleset{Col} and \ruleset{Graph-Fjords} to show both are \cclass{PSPACE}-complete.  For more information about the structure of each of these gadgets, the interested reader may reference \cite{DBLP:books/daglib/0023750}.

\subsection{Placement Games}

\emph{Placement games} on graphs are combinatorial rulesets played on graphs that fulfill all of these requirements:
\begin{itemize}
  \item Vertices are either marked or unmarked,
  \item A move for a player consists of marking one of the vertices, and
  \item Marks may never be moved or removed.\cite{BrownCHMMNS}
\end{itemize}

Since marks may never be moved or removed, the maximum number of plays made during a game is always $n$, where $n$ is the number of vertices in the graph.  Thus, an algorithm to compute the winner of any placement game needs only a polynomial amount of space; all placement games are in \cclass{PSPACE}.

The next three sections describe the three placement games considered in this paper: \ruleset{Col}, \ruleset{Graph-NoGo}, and \ruleset{Graph-Fjords}.

\subsection{\ruleset{Col}}

\ruleset{Col} is a partisan placement game where Left and Right alternately paint vertices with their color (\blue\ and \red) with the restriction that two neighboring vertices may not have the same color.  Thus a single turn consists of painting an uncolored vertex not adjacent to another vertex of the player's color.  A more formal definition is given in \cref{def:col}.

\begin{definition}[\ruleset{Col}]
  \label{def:col}
  \ruleset{Col} is a ruleset played on any graph, $G = (V, E, c)$, where $c$ is a coloring of vertices $c: V \rightarrow \{\blue, \red, \uncolored\}$ such that $\forall (v, w) \in E: $ either $c(v) \neq c(w)$ or $c(v) = c(w) = \uncolored$.  An option for player $A$ is a graph $G' = (V, E, c')$ with $\exists x \in V:$
  \begin{itemize}
    \item $\forall v \neq x \in V: c'(v) = c(v)$, and
    \item $c(x) = \uncolored$, and
    \item $c'(x) = A$'s color.
  \end{itemize}
\end{definition}

\ruleset{Col} was devised in 19XX.  Its computational complexity has remained an open problem since the 1970's.  We show that \ruleset{Col} is \cclass{PSPACE}-complete, even for planar graphs, in \cref{sec:colIsHard}.

\subsection{\ruleset{Graph-NoGo}}

\ruleset{Graph-NoGo} is a partisan placement game where Left and Right alternately paint vertices using their color (\black\ and \white) with the restriction that each connected component of one color must be adjacent to an \uncolored\ vertex.  A more formal definition is given in \cref{def:noGo}.

\begin{definition}[\ruleset{Graph-NoGo}]
  \label{def:noGo}
  \ruleset{Graph-NoGo} is a ruleset played on any graph, $G = (V, E, c)$, where $c$ is a coloring of vertices $c: V \rightarrow \{\black, \white, \uncolored\}$ such that $\forall$ connected single-color component $C \subseteq V: \exists (v,w) \in E$ where:
  \begin{itemize}
    \item $v \in C$, and
    \item $w \notin C$, and
    \item $c(w) = \uncolored$.
  \end{itemize}
  An option for player $A$ is a graph $G' = (V, E, c')$, where $\exists x \in V:$
  \begin{itemize}
    \item $\forall v \neq x \in V: c'(v) = c(v)$, and
    \item $c(x) = \uncolored$, and
    \item $c'(x) = A$'s color.
  \end{itemize} 
  and the same above property holds true, but using the new coloring $c'$: $\forall$ connected single-color (as colored by $c'$) component $C' \subseteq V: \exists (v, w) \in E$ where:
  \begin{itemize}
    \item $v \in C'$, and
    \item $w \notin C'$, and
    \item $c'(w) = \uncolored$.
  \end{itemize}
\end{definition}

\ruleset{NoGo} is the well-known version of \ruleset{Graph-NoGo} played specifically on a grid-graph.  Our hardness proof applies only to the more general \ruleset{Graph-NoGo}; the computational complexity of the grid version remains an open problem.

\ruleset{NoGo} is itself a non-loopy \ruleset{Go} variant where capturing moves are not allowed. (In \ruleset{Go}, uncolored vertices adjacent to a connected component of one color are known as liberties.  \ruleset{Graph-NoGo} enforces that a liberty must always exist for each connected component.)  Resolving the computational complexity has been considered an open problem since 2011 when a tournament was played among combinatorial game theorists at the Banff International Research Station.  

\subsection{\ruleset{Graph-Fjords}}

\ruleset{Graph-Fjords} is a partisan placement game where Left and Right alternately paint vertices using their color (\black\ and \white) with the restriction that the newly-painted vertex must be adjacent to a vertex already painted that player's color.  Note that the initial position must contain both colored and uncolored vertices in order to have any options.

\begin{definition}[\ruleset{Graph-Fjords}]
  \label{def:fjords}
  \ruleset{Fjords} is a ruleset played on any graph, $G = (V, E, c)$, where $c$ is a coloring of vertices $c: V \rightarrow \{\black, \white, \uncolored\}$.  An option for player $A$ is a graph $G' = (V, E, c')$ where $\exists (x, y) \in E:$
  \begin{itemize}
    \item $\forall v \neq x \in V: c'(v) = c(v)$, and
    \item $c(x) = \uncolored$, and
    \item $c'(x) = A$'s color, and
    \item $c(y) = A$'s color.
  \end{itemize}
\end{definition}

\ruleset{Graph-Fjords} is the generalized version of \ruleset{Fjords}, which is played on a hexagonal grid with some edges and vertices removed.  In the published version of \ruleset{Fjords}, the initial configuration is generated through a randomized process, but once that is complete the remainder of the game (as described here) is strictly combinatorial.  This paper only solves the general case; the computational complexity of \ruleset{Fjords} remains an open problem.

\section{\ruleset{Col} is \cclass{PSPACE}-complete}
\label{sec:colIsHard}

Next we show that \ruleset{Col} is hard.

\subsection{Reduction from Constraint Logic}

\begin{theorem}[\ruleset{Col} is Hard]
  \label{theorem:colHardness}
  Determining whether the next player in \ruleset{Col} has a winning strategy is \cclass{PSPACE}-hard.
\end{theorem}

  To complete this proof, we reduce from 2-player \ruleset{Constraint Logic} to \ruleset{Col}.  The resulting position's graph will have two separate components: a section only Left (\blue) can play on and a section where both players will play that consists of the gadgets presented here.
  
  The \blue-only component is a single star graph with $k$ rays (we'll specify $k$ later) and a single hub colored \red.
  
  This second section with the gadgets will be played in two stages.  In the first stage, the players first alternate choosing variable gadgets to play on, then \red\ finishes filling in the rest of the gadgets (\blue\ cannot play on the other gadgets) while \blue\ plays on their separate component.  When \red\ finishes correctly playing on the gadgets, they may be able to make a single last play on the goal gadget depending on whether they won the variable-selection phase.  If they incorrectly play on the gadgets, then that will cost them one or more moves and \blue\ will win.
  
  \begin{figure}[h!]
    \begin{center}
      \includegraphics[scale=.9]{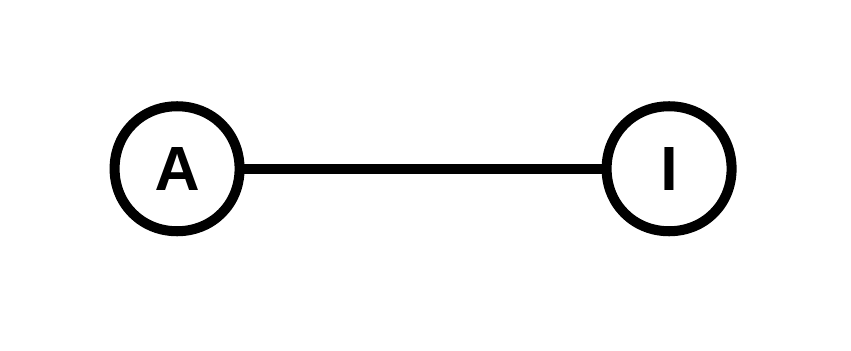}
	\end{center}
    \caption{The \ruleset{Col} gadget for a \ruleset{Constraint Logic}-Variable.  A stands for ``activated''; I stands for ``inactive''.}
    \label{fig:edge}
  \end{figure}
  
  The basic gadget represents the orientable \ruleset{Constraint Logic} edge (figure \ref{fig:edge}).  This is modeled by two connected vertices.  If Red colors the `A' or `I' vertex, this represents activating or leaving the edge inactive, respectively.  All of the other gadgets use these edges as either inputs or outputs (or both).  Red needs to play in each of the edges (and other places) in order to win the game.
  
  The remaining gadgets connect the edges to each other.  In order to suffice for \cclass{PSPACE} hardness, we need only to implement the gadgets to represent \ruleset{Constraint Logic} vertices that represent variables, goal-edges, splitters, path choice, and AND and OR gates.
  
  In order to complete the reduction, we need to show how to create \ruleset{Col} gadgets from each of the relevant \ruleset{Constraint Logic} gadgets: Variable, Output-Edge, Or, And, Choice, and Split.
  
  \begin{figure}[h!!]
    \begin{center}
      \includegraphics[scale=.9]{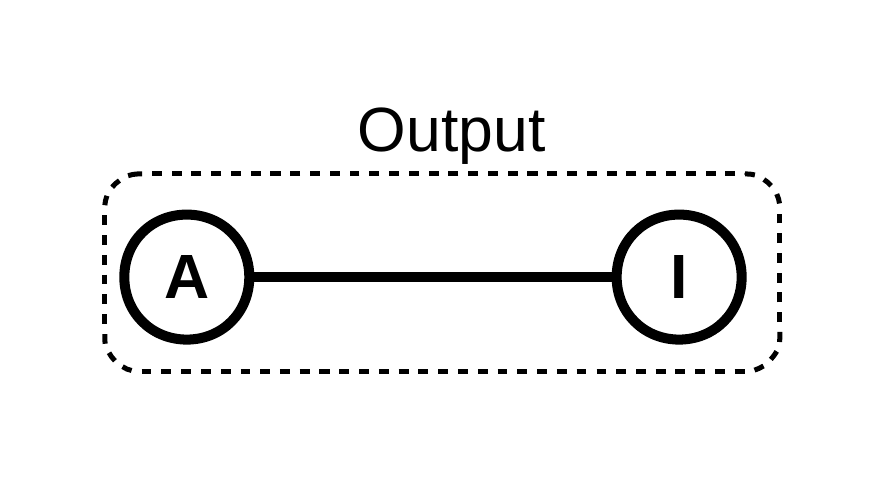}
	\end{center}
    \caption{The \ruleset{Col} gadget for a \ruleset{Constraint Logic}-Variable.}
    \label{fig:variable}
  \end{figure}
  
  
  \begin{figure}[h!]
    \begin{center}
      \includegraphics[scale=.8]{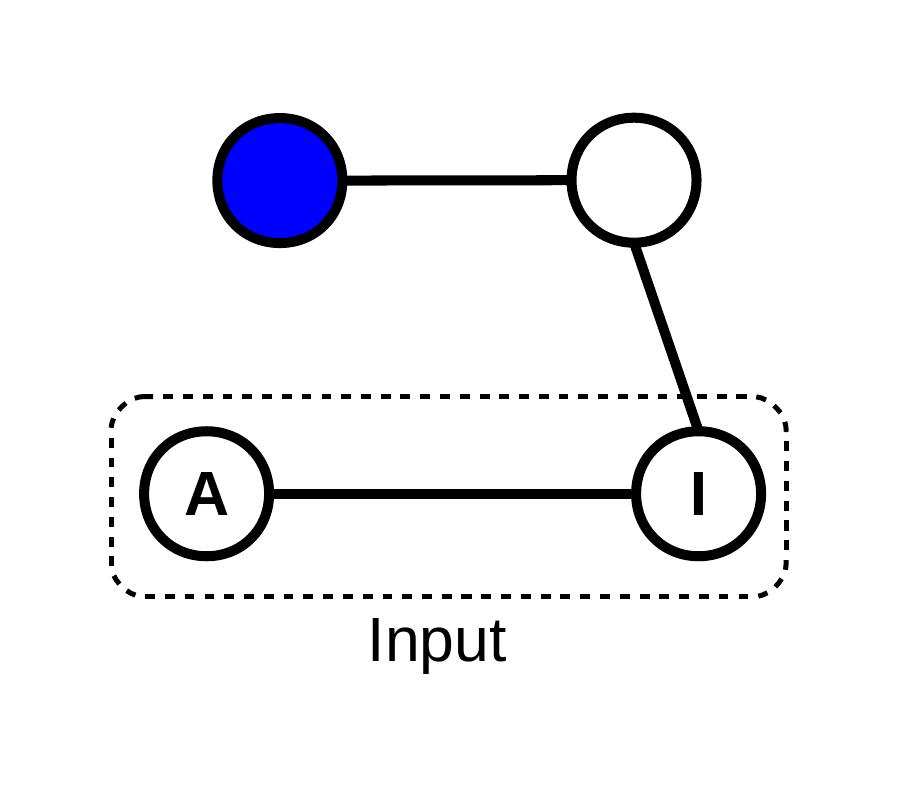}
	\end{center}
    \caption{The \ruleset{Col} gadget for a \ruleset{Constraint Logic}-Goal-Edge.  The only way to color the unmarked vertex is if the edge below is activated.}
    \label{fig:goalEdge}
  \end{figure}
  
  The Goal-Edge gadget includes a pair of adjacent vertices, which are connected to an edge pair.  (See figure \ref{fig:goalEdge}.)  One of the vertices is colored blue, while the other remains uncolored and is adjacent to the inactive vertex in the input variable.  This represents the edge that needs to be activated in order for \red\ to win the game.  
  
  \begin{figure}[h!]
    \begin{center}
      \includegraphics[scale=.5]{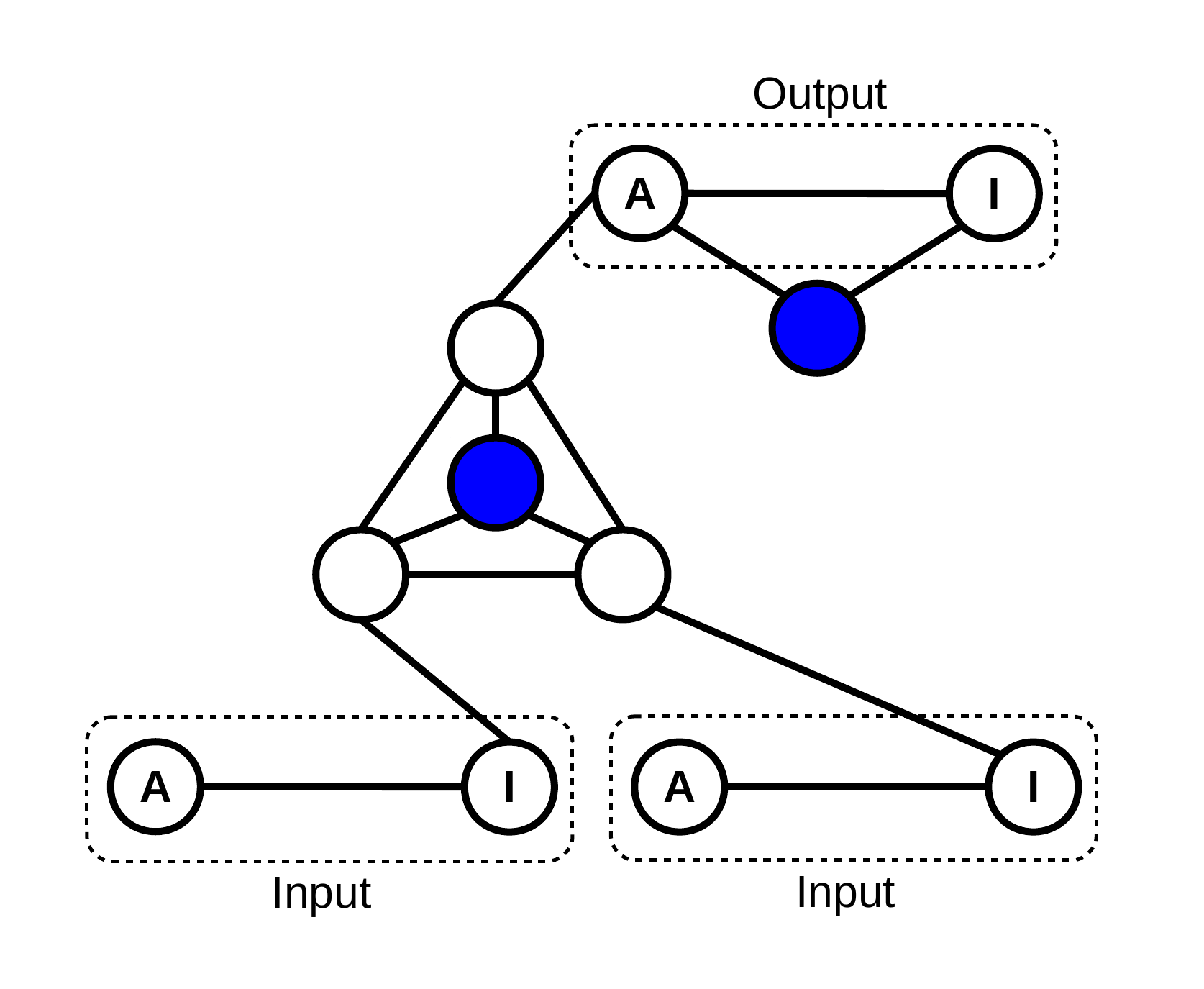}
	\end{center}
	\caption{The \ruleset{Col} gadget for a \ruleset{Constraint Logic}-Or.  If either of the inputs is activated, then the output can be activated and the inner structure can also be played.}
    \label{fig:or}
  \end{figure}
  
  The Or gadget, illustrated in figure \ref{fig:or}, connects three edge gadgets (two input and one output) to a 4-clique.  In the new clique, one vertex is colored blue, another is connected to the active vertex of the output variable, and the remaining two are connected, one each, to the inactive vertices of the inputs.  If either of the inputs is activated, then the output can be activated and the clique can also be played.  (\red\ will need to play in each of these inner cliques to have a chance to win.)  Otherwise, the clique can be played on only if the output variable is inactive.  
  
  \begin{figure}[h!]
    \begin{center}
      \includegraphics[scale=.5]{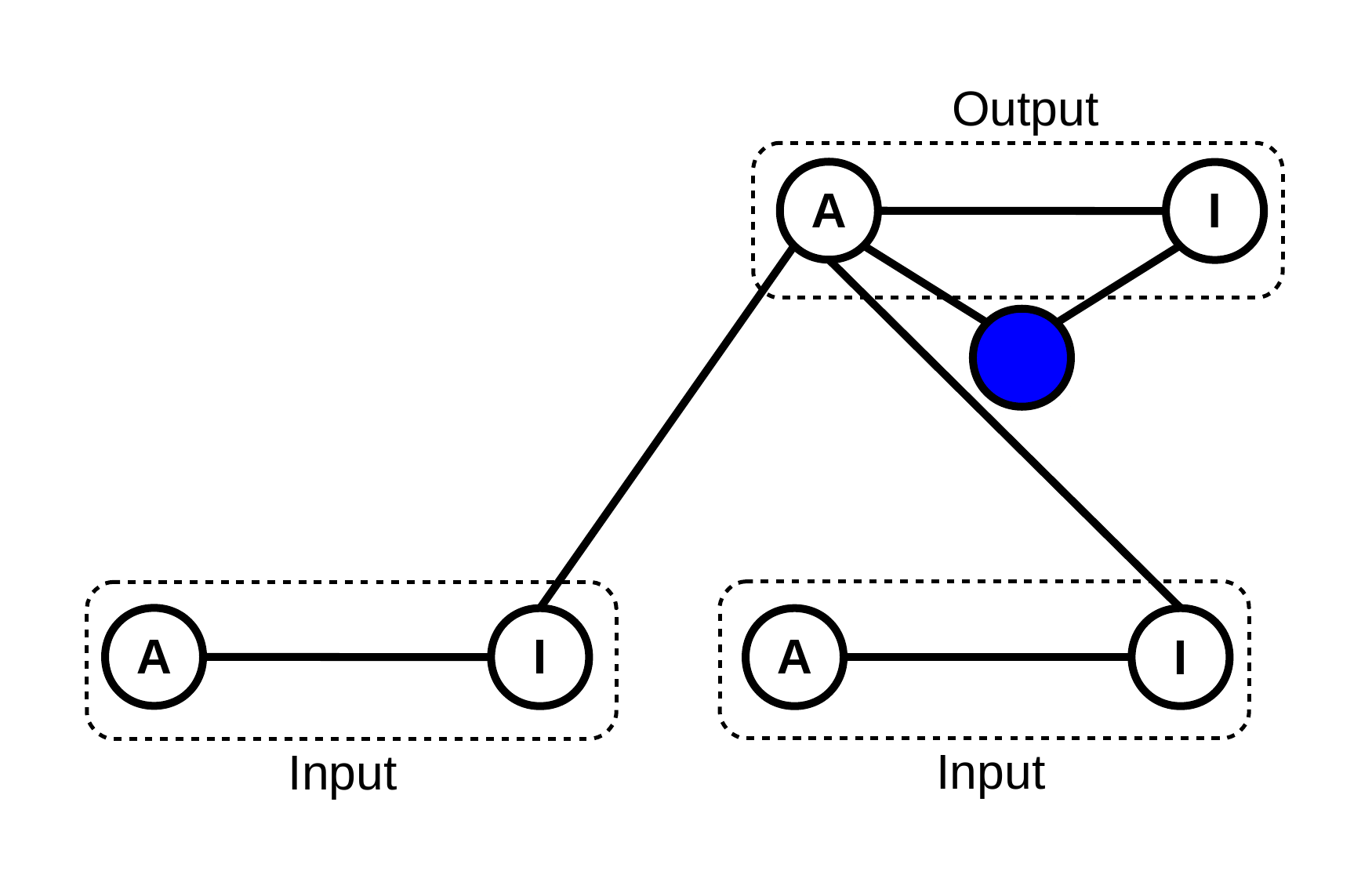}
	\end{center}
	\caption{The \ruleset{Col} gadget for a \ruleset{Constraint Logic}-And.}
    \label{fig:and}
  \end{figure}
  
  The And gadget, shown in figure \ref{fig:and}, also connects two input edge gadgets to an output.  The active vertex of the output is connected to each input inactive vertex.  In order to activate the output variable, both of the two input variables must be activated.
  
  \begin{figure}[h!]
    \begin{center}
      \includegraphics[scale=.6]{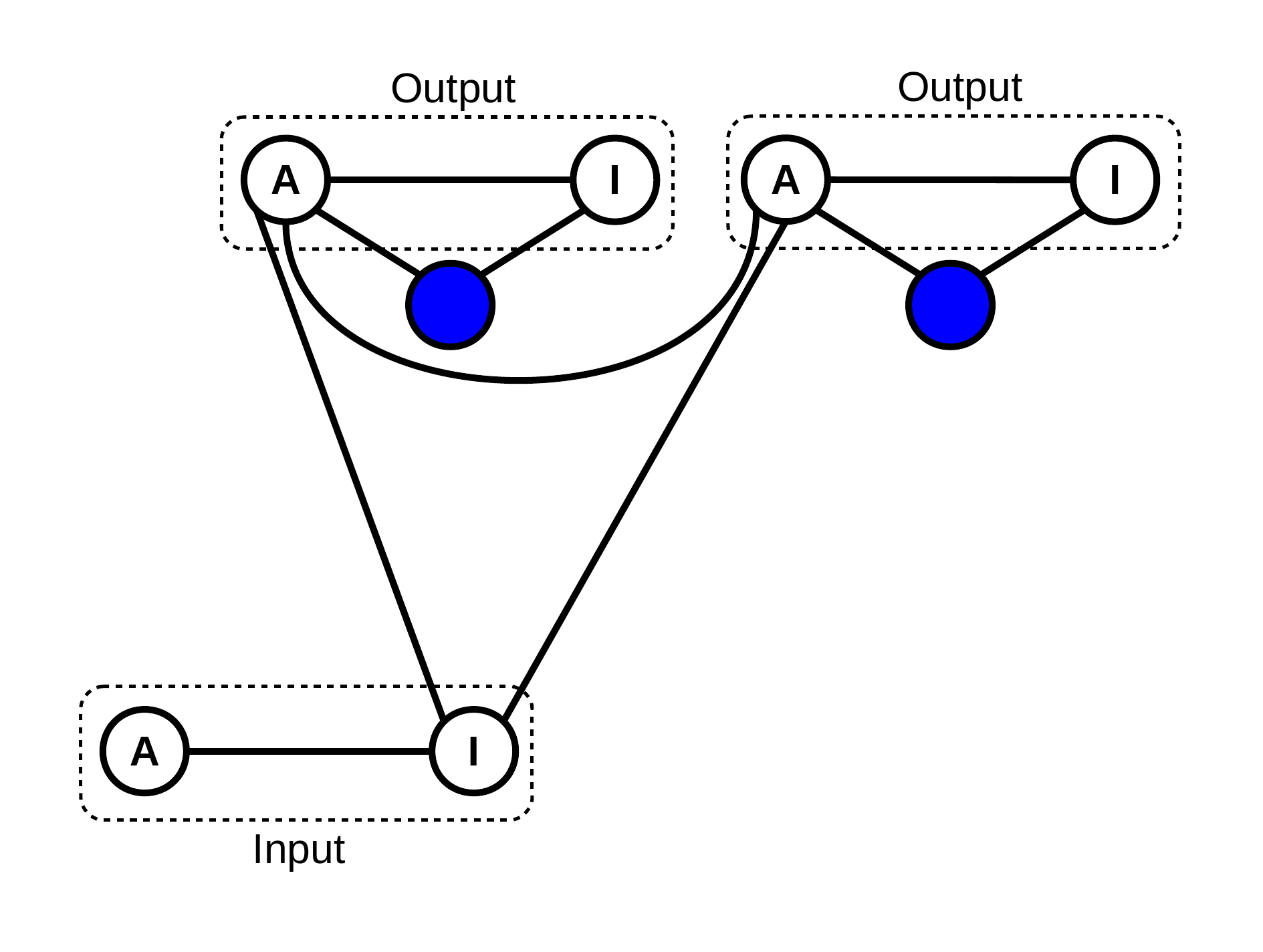}
	\end{center}
	\caption{The \ruleset{Col} gadget for a \ruleset{Constraint Logic}-Choice.}
    \label{fig:choice}
  \end{figure}
  
  The Choice gadget, figure \ref{fig:choice}, connects one input edge gadget to two outputs.  The active vertices of the outputs are adjacent to each other and both connected to the inactive vertex of the input.  If the input is chosen to be active, then one of the outputs may be activated.  (They cannot both be activated, because the active vertices are adjacent.)
  
  \begin{figure}[h!]
    \begin{center}
      \includegraphics[scale=.6]{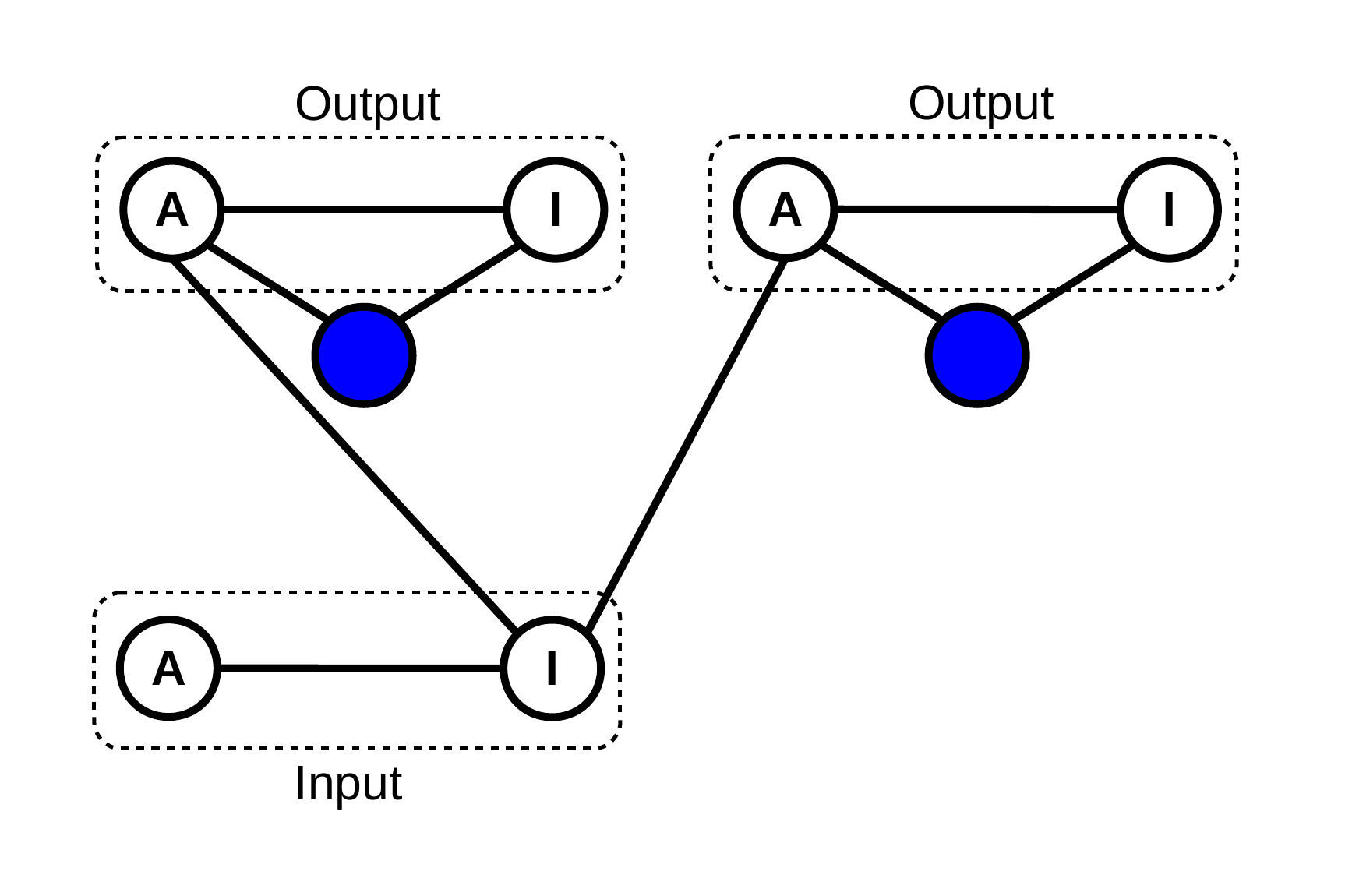}
	\end{center}
	\caption{The \ruleset{Col} gadget for a \ruleset{Constraint Logic}-Split.}
    \label{fig:split}
  \end{figure}
  
  The Split gadget, shown in figure \ref{fig:split}, is exactly the same as the Choice gadget, except that the active vertices of the outputs are not adjacent.  Thus, if the input variable is activated, \emph{both} outputs can be active.


\subsection{Planarity}

\begin{corollary}[Planar Hardness]
  Determining whether the next player in \ruleset{Col} on planar graphs has a winning strategy is \cclass{PSPACE}-hard.
\end{corollary}

\begin{proof}
  We can follow the same proof as in Theorem \ref{theorem:colHardness}.  Since \ruleset{constraint logic} is \cclass{PSPACE}-hard on planar graphs, and the reduction preserves planarity, planar \ruleset{Col} is also \cclass{PSPACE}-hard.
\end{proof}

\section{\ruleset{Graph-NoGo} is \cclass{PSPACE}-complete}

In this section, we prove that \ruleset{Graph-NoGo} is \cclass{PSPACE}-complete, even on planar graphs.  Since \ruleset{Graph-NoGo} is a graph placement game, it is already in \cclass{PSPACE}.  It remains to show that \ruleset{Graph-NoGo} is hard.

\subsection{\ruleset{Graph-NoGo} is \cclass{PSPACE}-hard}

To prove the hardness of \ruleset{Graph-NoGo}, we reduce from \ruleset{Col}, which was shown to be \cclass{PSPACE}-hard in \cref{theorem:colHardness}.  The reduction uses only two gadgets.

\begin{theorem}[\ruleset{Graph-NoGo} is \cclass{PSPACE}-hard]
  It is \cclass{PSPACE}-hard to determine whether the next player has a winning strategy in \ruleset{Graph-NoGo}.
\end{theorem}

\begin{proof}
  Let $G = (V, E, c)$ be an instance of \ruleset{Col}.  The result of the reduction on $G$ will be a \ruleset{Graph-NoGo} graph $G' = (V', E', c')$ defined as follows.
  
  \begin{figure}[h!]
    \begin{center}\includegraphics[scale = .6]{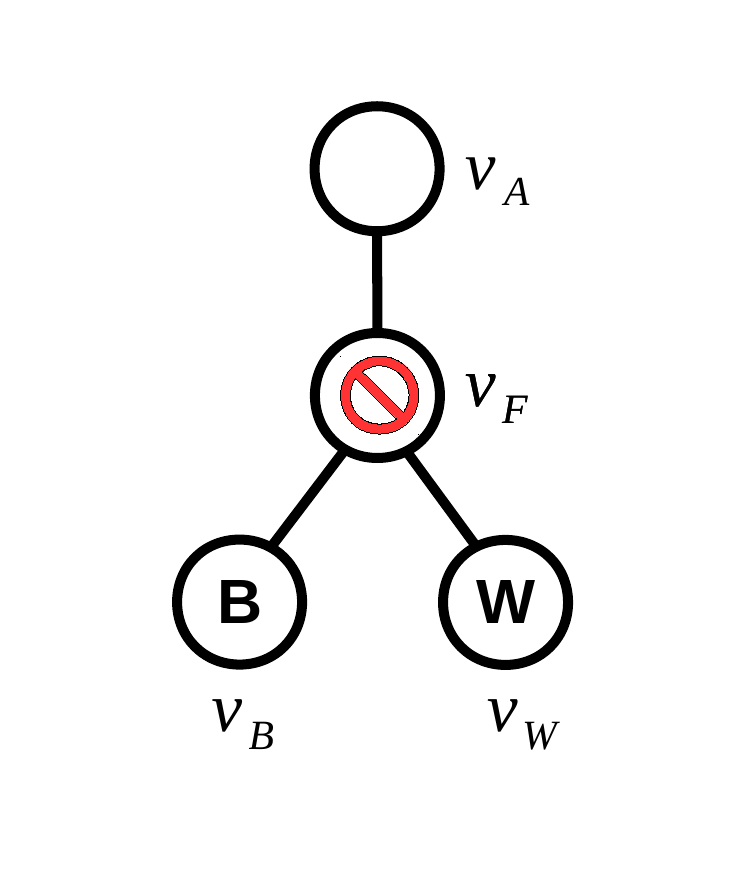}\end{center}
    \caption{The \ruleset{Graph-NoGo} gadget for each \ruleset{Col} vertex $v$.  $v_F$ cannot be colored and the component connected to $v_A$ will always have a liberty.}
    \label{fig:nogoVertex}
  \end{figure}
  
  For each vertex, $v$, in $V$, we include four vertices in $V'$: $v_B$, a vertex painted \black; $v_W$, a vertex painted \white; $v_F$, an \uncolored ``forbidden'' vertex that can never be painted; and $v_A$, a vertex that may be painted.  (Choosing to color $v_A$ corresponds to coloring \ruleset{Col} vertex $v$.)  We also include the three edges $(v_W, v_F), (v_B, v_F)$, and $(v_F, v_A)$ in $E'$.  This gadget is illustrated in \cref{fig:nogoVertex}.  $v_B$ and $v_W$ will never be adjacent to any vertices aside from $v_F$.  If $v_F$ were to be colored, either $v_B$ or $v_W$ would be a connected component without a liberty.  Thus, $v_F$ can never be colored, and will always be a liberty for $v_A$.

\begin{figure}[h!]
\begin{center}\includegraphics[scale = .6]{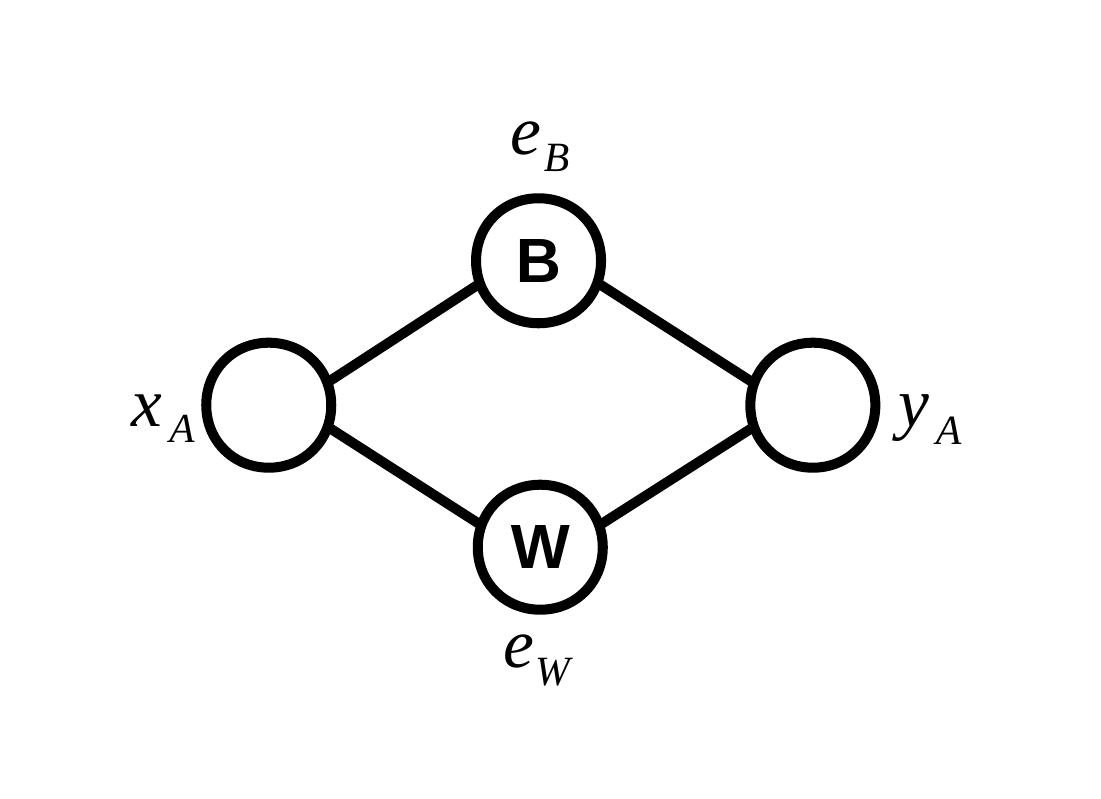}\end{center}
\caption{The \ruleset{Graph-NoGo} gadget for each \ruleset{Col} edge $e = (x, y)$.  $x_A$ and $y_A$ cannot be painted the same color because one of the two other edges.}
\label{fig:nogoEdge}
\end{figure}

  Also, for each edge, $e = (x, y)$, in $E$, we include two vertices and four edges in $G'$ to mimic the proper coloring rule in \ruleset{Col}.  This gadget is illustrated in \cref{fig:nogoEdge} and described more formally as follows.  Let $e_B$ and $e_W$ be vertices colored \black and \white, respectively, then include the four edges $(x_A, e_B)$, $(x_A, e_W)$, $(y_A, e_B)$, and $(y_A, e_W)$.  $e_B$ and $e_W$ will not be adjacent to any other vertices aside from $x_A$ and $y_A$.  If $x_A$ and $y_A$ have the same color, exactly one of the two new vertices $e_B$ or $e_W$ will be a connected component with no liberty.  Thus, they cannot both be painted the same color.

  Using the definitions above, we can formally define $V'$, $E'$, and $c'$.  First, $V'$:
  \begin{itemize}
    \item $V'_B = \{v_B \vdots v \in V \} \union \{e_B \vdots e \in E\}$,
    \item $V'_W = \{v_W \vdots v \in V \} \union \{e_W \vdots e \in E\}$,
    \item $V'_F = \{v_F \vdots v \in V \}$,
    \item $V'_A = \{v_A \vdots v \in V \}$, so
    \item $V' = V'_B \union V'_W \union V'_F \union V'_A$.
  \end{itemize}
  
  And $E'$:
  \begin{itemize}
    \item $E'_V = \Union_{v \in V} \{(v_B, v_F), (v_W, v_F), (v_F, v_A)\}$,
    \item $E'_E = \Union_{e = (x, y) \in E} \{(x_A, e_B), (x_A, e_W), (y_A, e_B), (y_A, e_W)\}$, so
    \item $E' = E'_V \union E'_E$.
  \end{itemize}
  
  And $c'$:
  \begin{itemize}
    \item $\forall v \in V'_B: c'(v) = \black$,
    \item $\forall v \in V'_W: c'(v) = \white$,
    \item $\forall v \in V'_F: c'(v) = \uncolored$,
    \item $\forall v \in V'_A: $ 
    \begin{itemize}
      \item if $c(v) = \uncolored$, then $c'(v) = \uncolored$
      \item if $c(v) = \blue$, then $c'(v) = \black$, and
      \item if $c(v) = \red$, then $c'(v) = \white$. 
    \end{itemize}
  \end{itemize}
  
  Each possible \ruleset{Graph-NoGo} move on $v_A \in V'$ for player $A$ cooresponds exactly to a \ruleset{Col} move on $v \in V$ for the same player, so the game trees and strategies for each are exactly the same.  Thus, \ruleset{Graph-NoGo} is also \cclass{PSPACE}-hard.
\end{proof}

\begin{corollary}[\ruleset{Graph-NoGo} is \cclass{PSPACE}-complete.]
  Determining whether the next player has a winning strategy in \ruleset{Graph-NoGo} is \cclass{PSPACE}-complete.
\end{corollary}

\begin{proof}
  Since \ruleset{Graph-NoGo} is a placement game, it is in \cclass{PSPACE}.  Since it is both in \cclass{PSPACE} and \cclass{PSPACE}-hard, it is \cclass{PSPACE}-complete.
\end{proof}

\begin{corollary}[\ruleset{Planar-NoGo} is \cclass{PSPACE}-complete]
  Determining whether the next player has a winning strategy in \ruleset{Graph-NoGo} is \cclass{PSPACE}-complete when played on planar graphs.
\end{corollary}

\begin{proof}
  Since \ruleset{Col} is \cclass{PSPACE}-complete on planar graphs, and the reduction gadgets do not require any additional crossing edges, the resulting \ruleset{Graph-NoGo} board is also planar.  Thus, \ruleset{Graph-NoGo} is also \cclass{PSPACE}-complete on planar graphs.
\end{proof}

\section{\ruleset{Graph-Fjords} is \cclass{PSPACE}-complete}

\subsection{\ruleset{Graph-Fjords} is \cclass{PSPACE}-hard}

For \ruleset{Graph-Fjords}, we will also reduce from \btcl: for the reduction we must cover variable, split, choice, and, and or gadgets, as well as the final victory gadget.  Assuming the \btcl\ position is the result of a reduction from \ruleset{POS-CNF}, thus players should choose the variables first, then play on the rest of the gadgets.

\subsubsection{\ruleset{Graph-Fjords} variables}

In order to enforce this, we'll add incentives to the variable gadgets so that they are all played before the other gadgets.  Each gadget for a variable $x_i$ will consist of a single \ruleset{Graph-Fjords} vertex, $v_i$, connected to both a black and white vertex as well as a separate clique of size $t$, $K_t$, as shown in \cref{fig:fjordsVariable}.  For further gadgets that use $x_i$, $v_i$ will be included as the active input. 

\begin{figure}[h!]
    \begin{center}\includegraphics[scale = .6]{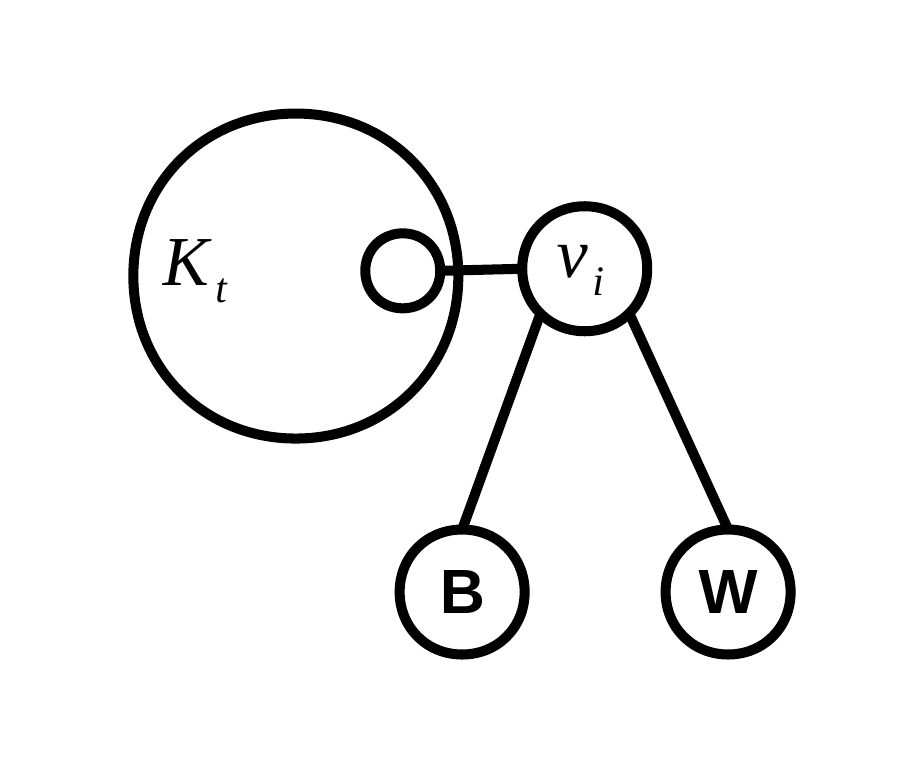}\end{center}
    \caption{The \ruleset{Graph-Fjords} variable gadget.  The left figure shows the $k$-clique; the right figure shows our replacement symbol used throughout the figures.}
    \label{fig:fjordsVariable}
\end{figure}

Other black and white-colored vertices exist in the remainder of the Fjords graph, but $t$ will be large enough so that all variables will be chosen first.

After all $n$ variables are chosen, we want each player to have cached $\ceiling{\frac{n}{2}t}$ ``moves'' in unclaimed vertices accessible only to them.  In order to make this happen, if $n$ is odd, we add an extra "dummy" variable gadget so that this balances out.

\subsubsection{\ruleset{Graph-Fjords} Input/Output Pairs}

As with the \ruleset{Col} reduction in \ref{sec:colIsHard}, there are Input/Output vertex pairs (with an active vertex and inactive vertex) that are outputs to some gadgets and inputs to others. As in \cref{fig:fjordsIO}, all figures here will have the active vertex on the left an the inactive on the right.

\begin{figure}[h!]
    \begin{center}\includegraphics[scale = .6]{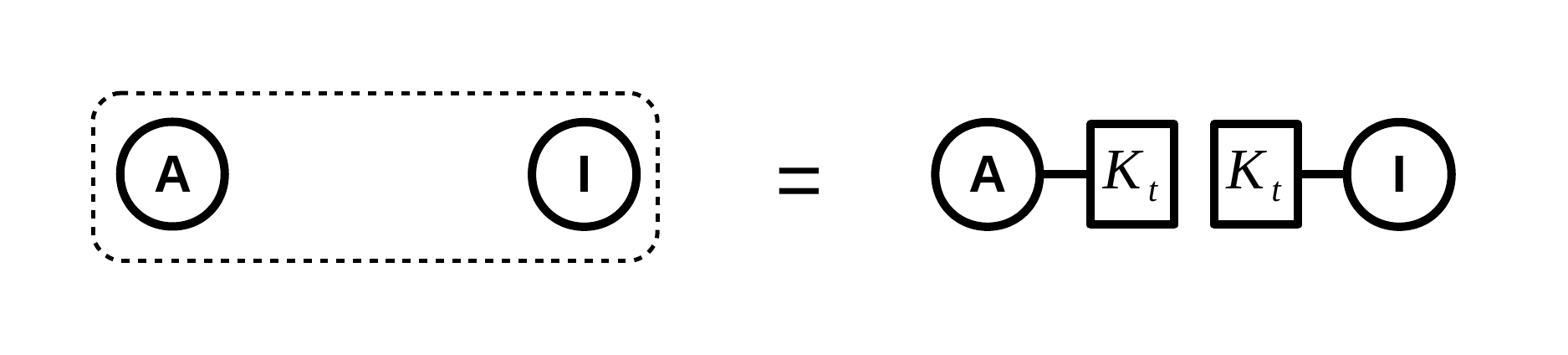}\end{center}
    \caption{The \ruleset{Graph-Fjords} input-output vertex pair.  The pair is active if A is colored \black.}
    \label{fig:fjordsIO}
\end{figure}

Unlike the \ruleset{Col} reductions, both players play in all pairs (aside from variables, which don't have an inactive vertex).  The pair is active or inactive depending on where Left plays.  If the active vertex is \black (and the inactive is \white), then the pair is active.  Otherwise, the pair is inactive.

\subsubsection{\ruleset{Graph-Fjords} Goal}

At the other end of the Fjords reduction is a single goal vertex, shown in \cref{fig:fjordsGoal}.  If the source \btcl\ is winnable by Left, then in the \ruleset{Graph-Fjords} position, Left will be able to move last by claiming this vertex because they activated the final input-output pair.

\begin{figure}[h!]
    \begin{center}\includegraphics[scale = .6]{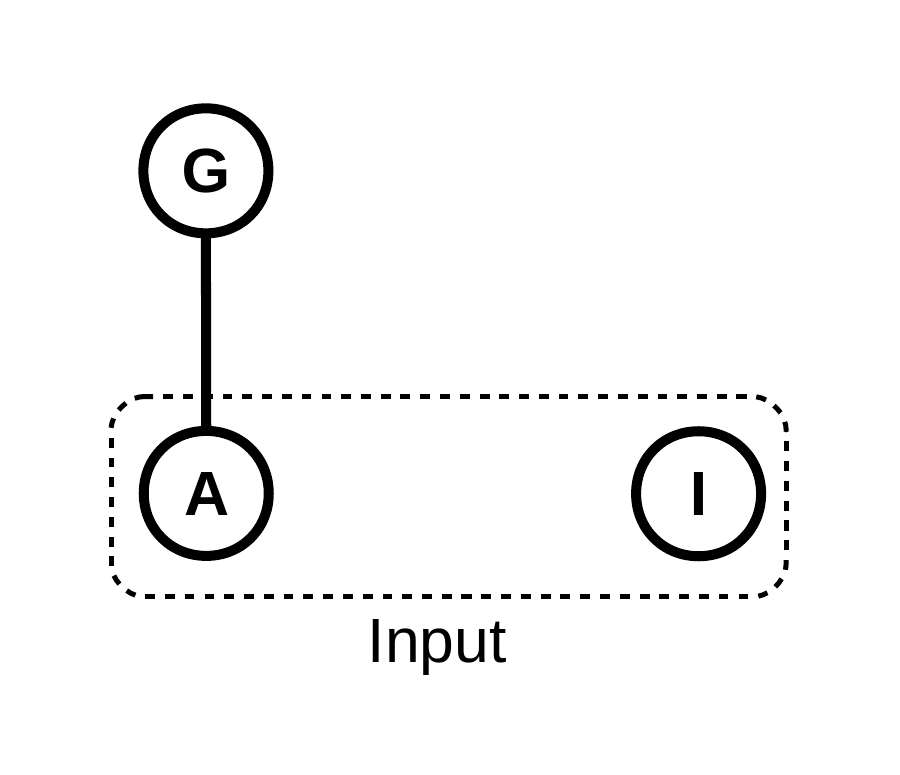}\end{center}
    \caption{The \ruleset{Graph-Fjords} goal gadget.}
    \label{fig:fjordsGoal}
\end{figure}

\subsubsection{\ruleset{Graph-Fjords} Intermediate Gadgets}

Between the variables and the goal, a series of gadgets will be laid out to implement the formula.  We enforce order on these by including decreasing incentives.  For gadget $X$ with an output that is an input to gadget $Y$, $X$ will have a higher incentive, $t$.  Just as with the variable gadgets, this will be realized as cliques $K_t$ attached to the gadget.

Since the source formula from \ruleset{POS CNF} has no negations, it is always better for Left to choose to make an output active over inactive.  (The same is true for Right.)

add more here?

In order to complete the reduction, we must include gadgets for the choice, split, or, and and \btcl gadgets.  

\subsubsection{\ruleset{Graph-Fjords} Choice Gadget}

The choice gadget allows Left to choose between one of two outputs, but only if the output is active.  (See \cref{fig:fjordsChoice} for an example.)  Since Left goes first, if the input is active, they will be able to choose between the active outputs.  Right will then respond by claiming the other active output.  Left and Right will then trade turns choosing the remaining inactive inputs. 

\begin{figure}[h!]
    \begin{center}\includegraphics[scale = .6]{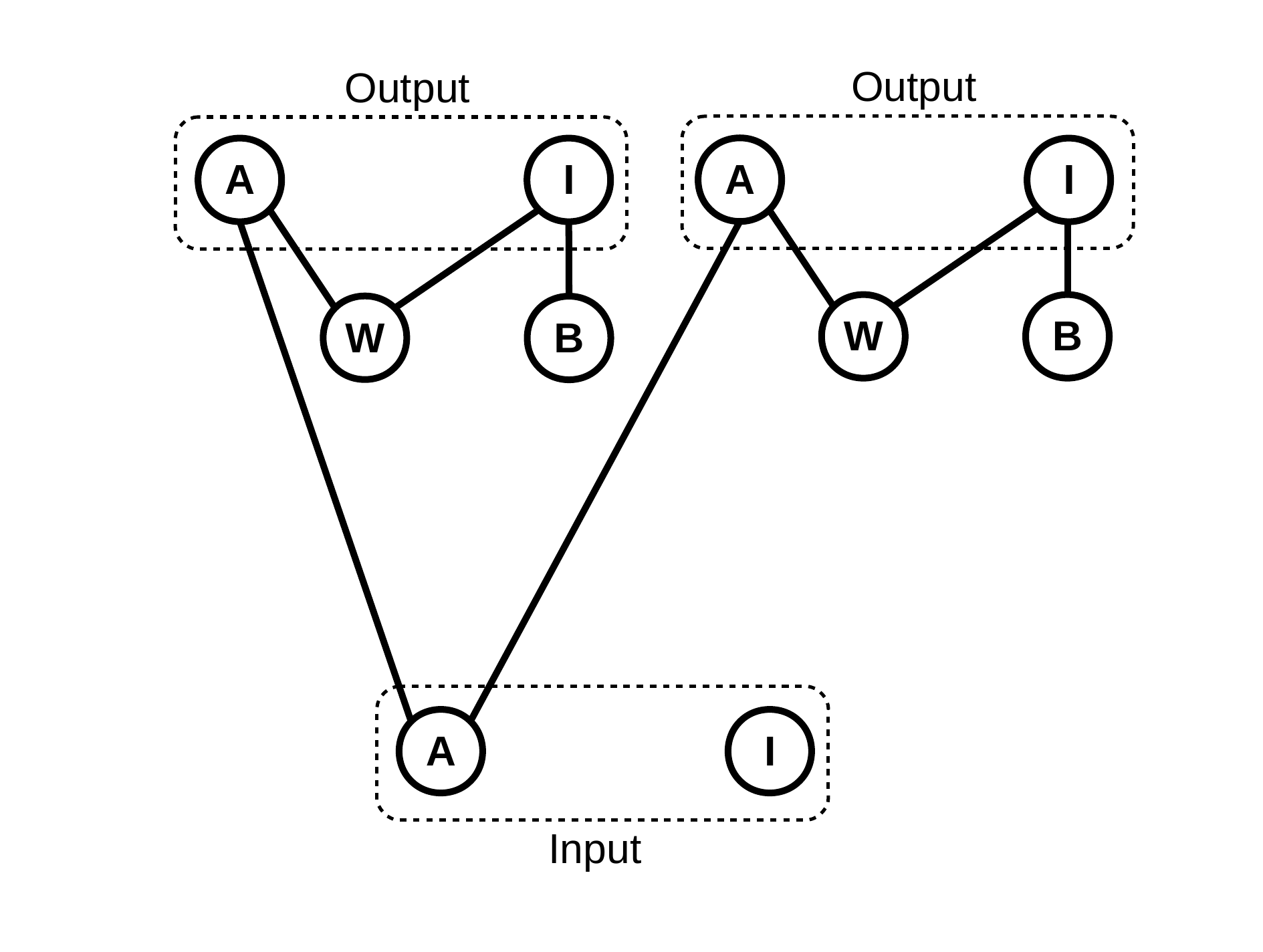}\end{center}
    \caption{The \ruleset{Graph-Fjords} choice gadget.}
    \label{fig:fjordsChoice}
\end{figure}

\subsubsection{\ruleset{Graph-Fjords} Split Gadget}

The split gadget copies the input pair into two output pairs and is shown in \cref{fig:fjordsSplit}.  In this gadget, whichever color is on the active input will get to choose both inactive outputs.

\begin{figure}[h!]
    \begin{center}\includegraphics[scale = .6]{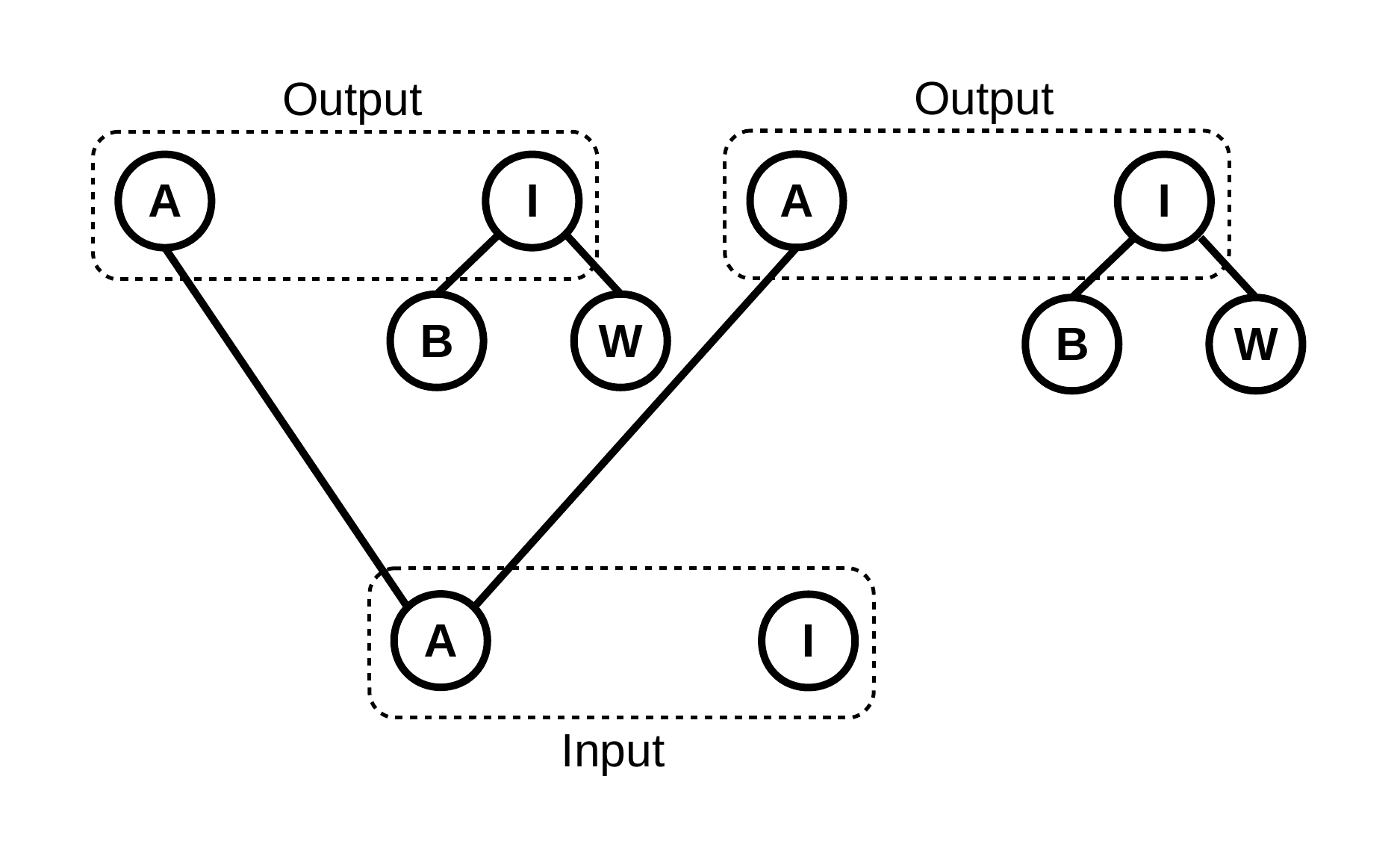}\end{center}
    \caption{The \ruleset{Graph-Fjords} split gadget.}
    \label{fig:fjordsSplit}
\end{figure}

\subsubsection{\ruleset{Graph-Fjords} Or Gadget}

The or gadget, shown in \cref{fig:fjordsOr}, works as expected: if either of the inputs is active, then Left may move to activate the output.  Otherwise, the output must be inactive.

\begin{figure}[h!]
    \begin{center}\includegraphics[scale = .6]{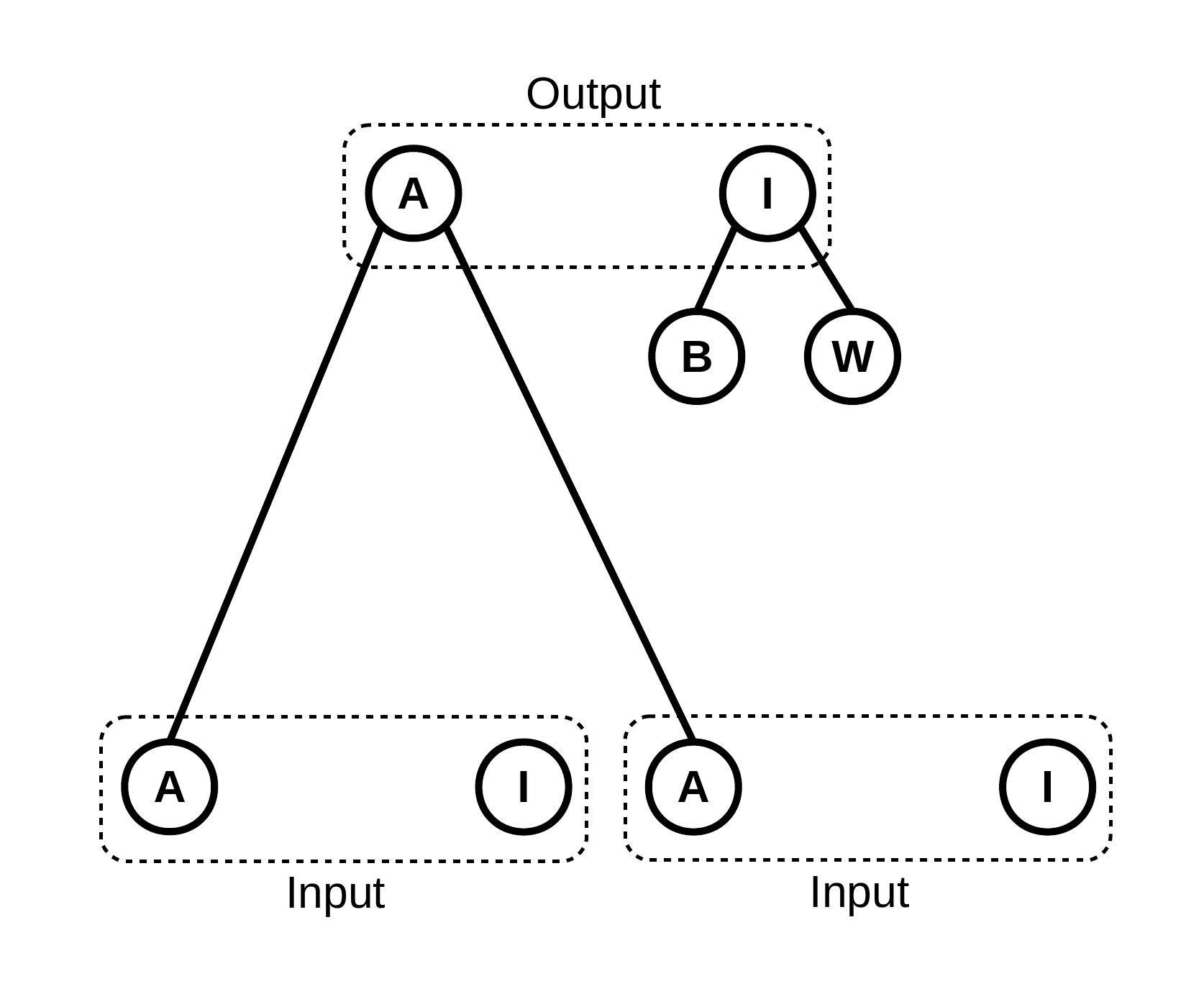}\end{center}
    \caption{The \ruleset{Graph-Fjords} or gadget.}
    \label{fig:fjordsOr}
\end{figure}

\subsubsection{\ruleset{Graph-Fjords} And Gadget}

The and gadget, shown in \cref{fig:fjordsAnd}, is the most surprising as it is not symmetric.  In this gadget, both inputs must be active in order for it to be worth Left's turn to color the active output vertex \black.  Clearly, if the left-hand input is inactive, Left won't have the chance to color that active input.  It remains to consider the situation where the left-hand input is active and the right-hand input is inactive.

\begin{figure}[h!]
    \begin{center}\includegraphics[scale = .6]{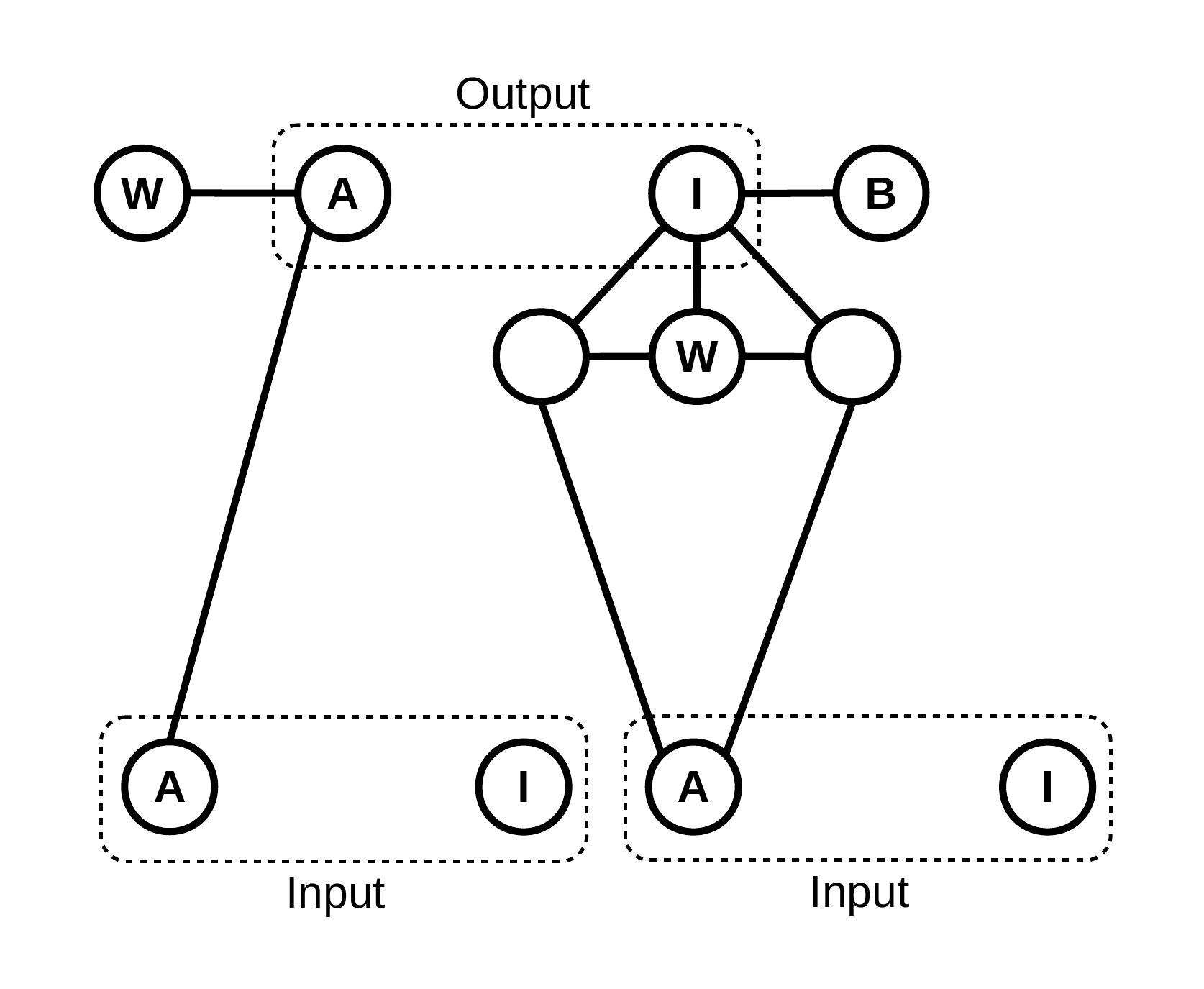}\end{center}
    \caption{The \ruleset{Graph-Fjords} and gadget.}
    \label{fig:fjordsAnd}
\end{figure}

In this case, if Left decides to activate the output, Right can play on the inactive vertex.  Now Right has cut off two extra vertices that can be taken later.  Even if Left plays the rest of the game perfectly, Right will win.  Right can never be cut off from these two spaces, but Left must be sure to occupy either the inactive output or the right-hand active input.

\subsubsection{Setting the Incentives}

In order to finish putting the reduction together, we need to prescribe incentives on each gadget so that it is played after its inputs are chosen.  Let $m$ be the number of non-variable gadgets.  Then:

\begin{itemize}
  \item  Those $m$ gadgets will have incentives of $0, 3, 6, 9, \ldots, 3m - 3$, starting with the goal and building back from that.  These are the sizes of the cliques connected to each output vertex.  There are no cycles in the flow of the reduction, so we can order the gadgets this way.
  \item  Since each gadget is worth at least three more ``points'' than any gadget off it's output(s), it is never worthwhile to play out of order: even for And gadgets where an extra two vertices could be gained by Right playing out of order on the inactive vertex.
  \item  The variables will all have incentive $3m$ so they are all worth playing at before any of the other gadgets.
\end{itemize}

\begin{theorem}[\ruleset{Graph-Fjords} is \cclass{PSPACE}-hard]
  \label{theorem:fjordsIsHard}
  Determining whether the next play has a winning strategy in \ruleset{Graph-Fjords} is \cclass{PSPACE}-hard.
\end{theorem}

\begin{proof}
  By using the scheme and gadgets described above, we can always reduce from \ruleset{B2CL} to \ruleset{Graph-Fjords}.
\end{proof}

\begin{corollary}[\ruleset{Graph-Fjords} is \cclass{PSPACE}-complete]
  Determining whether the next player has a winning strategy in \ruleset{Graph-Fjords} is \cclass{PSPACE}-complete.
\end{corollary}

\begin{proof}
  We can determine the winner of any placement game in \cclass{PSPACE}. Thus, by \cref{theorem:fjordsIsHard}, \ruleset{Graph-Fjords} is \cclass{PSPACE}-complete.
\end{proof}

\subsection{Planar Fjords}

\begin{corollary}[\ruleset{Planar-Fjords} is \cclass{PSPACE}-complete]
  Determining whether the next player has a winning strategy in \ruleset{Graph-Fjords} on planar graphs is \cclass{PSPACE}-complete.
\end{corollary}

\begin{proof}
  Since the reduction from \btcl\ preserves planarity, \ruleset{Graph-Fjords} is also \cclass{PSPACE}-hard on planar graphs.  Since it is in \cclass{PSPACE}, it is also \cclass{PSPACE}-complete.
\end{proof}

\section{Conclusions}

The main results of this paper are the computational hardness results for the three placement games \ruleset{Col}, \ruleset{Graph-NoGo}, and \ruleset{Graph-Fjords}: all three are \cclass{PSPACE}-complete.  Moreover, \ruleset{Col} and \ruleset{Graph-NoGo} 
are \cclass{PSPACE}-complete on planar graphs.  Even similar placement games such as \ruleset{Snort} and \ruleset{Node-Kayles} are only known to be \cclass{PSPACE}-complete on non-planar graphs, though those hardness results were known in the 1970's \cite{DBLP:journals/jcss/Schaefer78}.

The hardness of \ruleset{Col} is especially satisfying, as the problem has been open for decades without known progress.  Since the reduction proving hardness for \ruleset{Graph-NoGo} starts from \ruleset{Col}, we expect that \ruleset{Col} will be useful as the source for other reductions.

\section{Future Work}

There is still many placement games with unknown computational complexity.  In particular, it is still unknown whether \ruleset{NoGo} (\ruleset{Graph-NoGo} on grid-graphs) and \ruleset{Fjords} (\ruleset{Graph-Fjords} on subgraphs of a hexagonal grid) are computationally difficult.  Either of these would be a large improvement over the current result.

\begin{openProblem}
  Is \ruleset{NoGo} computationally hard?
\end{openProblem}

\begin{openProblem}
  Is \ruleset{Fjords} computationally hard?
\end{openProblem}

There are many other placement games that can be considered. 

\bibliographystyle{plain}

\end{document}